\documentclass[conference]{IEEEtran}
\IEEEoverridecommandlockouts
\usepackage{cite}
\usepackage{amsmath,amssymb,amsfonts}
\usepackage{algorithmic}
\usepackage[linesnumbered,ruled,vlined]{algorithm2e}
\SetKwProg{Fn}{Function}{:}{}
\usepackage{subfigure}
\usepackage{graphicx}
\usepackage{textcomp}
\usepackage{amsthm}
\newtheorem{theorem}{Theorem}[section]

\DeclareMathAlphabet{\mathbbold}{U}{bbold}{m}{n}
\usepackage{xcolor}
\usepackage{multirow}
\def\BibTeX{{\rm B\kern-.05em{\sc i\kern-.025em b}\kern-.08em
    T\kern-.1667em\lower.7ex\hbox{E}\kern-.125emX}}
\begin{document}
\title{Finite Horizon Multi-Agent Reinforcement Learning in Solving Optimal Control of State-Dependent Switched Systems\\
}
\author{Mi Zhou, Jiazhi Li, Masood Mortazavi, Ning Yan, and Chaouki Abdallah
  \thanks{Chaouki Abdallah is with the School of Electrical and Computer Engineering, Georgia Institute of Technology.
  Mi Zhou, Jiazhi Li, Masood Mortazavi, and Ning Yan are with Futurewei Technologies, San Jose, CA 30332.
  Email: {\tt\small mzhou2@futurewei.com, 
  jli5@futurewei.com,
  yan.ningyan@futurewei.com, masoodmortazavi@gmail.com, ccabdallah@gatech.edu}.}}
\maketitle

\begin{abstract}
In this article, a \underline{S}tate-dependent \underline{M}ulti-\underline{A}gent \underline{D}eep \underline{D}eterministic \underline{P}olicy \underline{G}radient (\textbf{SMADDPG}) method is proposed in order to learn an optimal control policy for regionally switched systems.
We observe good performance of this method and explain it in a rigorous mathematical language using some simplifying assumptions in order to motivate the ideas and to apply them to some canonical examples.
Using reinforcement learning, the performance of the switched learning-based multi-agent method is compared with the vanilla DDPG in two customized demonstrative environments with one and two-dimensional state spaces.
\end{abstract}

\begin{IEEEkeywords}
Reinforcement learning, state-dependent switched system, optimal control, Hamilton-Jacobi-Bellman equation, DDPG, approximation theory
\end{IEEEkeywords}

\section{Introduction}
Optimal control for state-dependent switched systems has been a hot topic in the optimal control community for a long time.
Applications of such systems can be found in fermentation processes \cite{Liu2014OptimalCO}, temperature control \cite{tc}, aerospace systems \cite{multiphase}, robots \cite{MITcheetah,Magnusrobots}, as well as natural systems \cite{snelllaw}.
There are different names for this problem, such as optimal control of differential systems with discontinuous right-hand side \cite{discrightside}, optimal control for piecewise smooth systems \cite{aeron}, optimal control for systems with isolated equality constraint \cite{optimalcontrolbook}, optimal control for systems with linear complementarity constraints, and optimal control for hybrid switched systems \cite{Azhmyakov}.
Despite the name, they all share the same properties.
\cite{discrightside} introduced several common state-dependent switched systems, such as a YO-YO model $\Ddot{x}+(1+u) \mathrm{sgn}(x) = 0$, a particle moving on a V-shaped configuration with model $\Ddot{x}+F\sin \alpha \cos\alpha \mathrm{sgn}x = 0$, a controlled system with Coulomb friction $\Ddot{x}+\mathrm{sgn}\dot x +x = u$, and an electrical relay system.
\cite{aeron} presented a Spring Loaded Inverted Pendulum (SLIP) model with massless legs for which, at different regions, the dynamics are different due to the position of the foot (single support and double support with legs will cause different dynamics of the system).
The theoretical foundation for this problem (hybrid minimum principle) has been well-built in \cite{Azhmyakov,ali2018,optimalcontrolbook,Shaikh, GELACC_Mi}.
The study of co-state was seen in \cite{costateDiscontinuous,economicmeaning,GELACC_Mi}.

However, if hybrid systems have regional dynamics or partitioned state space, when to switch and choosing which interface to switch make solving this problem a mixed integer nonlinear programming (MINLP) problem, for which, computational complexity is high and a globally optimal solution is hard to be reached.
Therefore, existing works focus on searching for a feasible and suboptimal solution.
Methods include graph search \cite{BBMPC}, nonlinear programming \cite{RUNGGERnonlinearProg}, dynamic programming \cite{RUNGGER2011}, model predictive control \cite{aeron}, and gradient decent \cite{multiregion}.
For example, the study of optimal control of multi-region state-dependent switched systems appeared in \cite{BBMPC} where the authors used branch-and-bound based model predictive control.
This combination of branch-and-bound and model predictive control still has a high computing demand.
\cite{RUNGGERnonlinearProg} proposed a numerical method by dynamic programming for solving this problem based on nonlinear programming by a full discretization of the state and input spaces.
In \cite{HBellman}, a hybrid Bellman Equation for systems with regional dynamics was proposed.
In this article, the switching interface was discretized and then dynamic programming was used at the high level to decide which region to switch to and which two states to be chosen as the boundary condition of a low-level optimal control problem.
Thus the problem is transformed into a dynamic programming problem between any two tuples $(t_s^j, \epsilon_s^j)$ and  $(t_s^i, \epsilon_s^i)$ where $\epsilon_s$ is the optimal switching state and $t_s$ is the optimal switching time.
The time complexity and space complexity of this algorithm are humongous considering the dynamic programming technique and the discretization. 
The accuracy of the final results also depended on the discretization precision.

With so many numerical methods for optimal control problems, none of them is a very general and convenient algorithm for state-dependent switched systems. 
As two important theoretical tools in optimal control communities, the Pontryagin minimum (maximum) principle and dynamic programming both have their own defects.
Pontryagin minimum principle requires that the system dynamics and stage cost be continuously differentiable which is not the case for the state-dependent switched systems.
Even with the development of the hybrid minimum principle, a flexible and user-friendly toolbox for such a problem is still not available.
And it is very easy for this method to get stuck into a local optimal.
The dynamic programming method is able to find a global optimal solution.
However, the Hamilton-Jacobi-Bellman equation is a partial differential equation that suffers from the curse of dimensionality using the traditional finite difference methods or level set methods.
Thus in this article, we borrow the idea of reinforcement learning and use neural networks to find an optimal solution.

Recently, there have been substantial works using reinforcement learning methods (RL) for switched systems.
Reinforcement learning algorithms such as deep deterministic policy gradient (DDPG), soft actor critic algorithms are widely studied from both theory and applications.
Most of them consider either the infinite horizon reinforcement learning problem \cite{Warren} or time-dependent switched systems \cite{RL-TD}.
In \cite{jingangZhao}, the authors considered a finite horizon reinforcement learning problem.
However, the system must be continuous and should have a control-affine form.

In this article, we first investigate finite horizon reinforcement learning for state-dependent switched systems.
Some theoretical foundations of finite horizon RL and state-dependent switched systems are given.
To solve the bias issue of gradient estimation, we then propose a state-dependent multiple-agent DDPG framework, named SMADDPG.
Two examples are provided to verify the efficiency of the proposed algorithm by comparing it with the vanilla DDPG algorithm.

This article is structured as follows: In Section \ref{sec:problem}, we introduce our problem as a hybrid optimal control problem.
In Section \ref{sec:error}, we present a finite horizon optimal control problem from the perspective of the Hamilton-Jacobi-Bellman equation and give an error analysis for approximate dynamic programming.
In Section \ref{sec:DDPG}, we introduce the framework of the SMADDPG algorithm.
Furthermore, we provide two multi-region state-dependent switched systems to verify our algorithm and compare it with the DDP in Section \ref{sec:simulation}.
Finally, Section \ref{sec:conclusion} concludes this article and presents some future works.
\section{Problem formulated}\label{sec:problem}
A general state-dependent switched system is defined as follows:
\begin{align}\label{eqn:realsys}
\dot x(t) = f_{q(t)}(x(t),u(t)),  \quad (x(t),u(t)) \in \mathcal{R}_{i}
\end{align}
where $x \in \mathbb{R}^{n_x}$ denotes the state, $u\in \mathbb{R}^{n_u}$ denotes the continuous control input, $q(t)$ is the index of sub-systems with which we can define a sequence of switching $q(t) = \{i|i\in (1, 2,..., m)\}$, and $f_{i}(x(t),u(t),t)$ is a continuously differentiable function in the region $\mathcal{R}_i$.
Then the optimal control problem is to find the optimal $u$ and switching time instant $\tau_i$ to minimize the following objective function:
\begin{align}\label{eqn:realJ}
J = \psi(x(t_f))+\sum_{i=1}^m\int_{\tau_i}^{\tau_{i+1}}L_{q(t)}(x(t),u(t)) {\rm d}t,
\end{align}
where stage cost $L_{i}(x,u)$ is also piecewise continuously differentiable functions.
$x_f$ denotes the terminal state and $\psi(x(t_f))$ is the terminal cost.

We make the following assumptions:
\begin{enumerate}
    \item There is no jump behavior of the state, i.e.,
\begin{align}
x(t)\in \mathcal{AC}([0,t_f ]; \mathbb{R}^{n_x})
\end{align}
where $\mathcal{AC}$ denotes absolute continuous function.
\item At each region, the system is controllable.
\item No Zeno behavior exists at the switching interface.
\item The optimal controller has bounded variations.
\end{enumerate}

Solving this problem is of significant difficulty especially when there are multiple regions.
What's worse, the nonlinear switching interface will make this problem harder to solve.
For example, in Fig. \ref{fig:scenario}, there are four regions and each region has different system dynamics.
Assuming an agent is going from $x_0$ at Region 1 to a target state $x_f$
at Region 3, the system can go through the trajectory $1\rightarrow 2 \rightarrow 3$.
It can also go from the trajectory of $1\rightarrow4\rightarrow3$.
It may also travel through $1\rightarrow 2 \rightarrow 1 \cdots \rightarrow 3$.
This model implies the combinatorial explosion of hybrid systems.
Without prior knowledge of how many switches will happen and what region it will go through, solving this problem is an NP-hard problem.
\begin{figure}[!htp]
    \centering
\includegraphics[width=\linewidth]{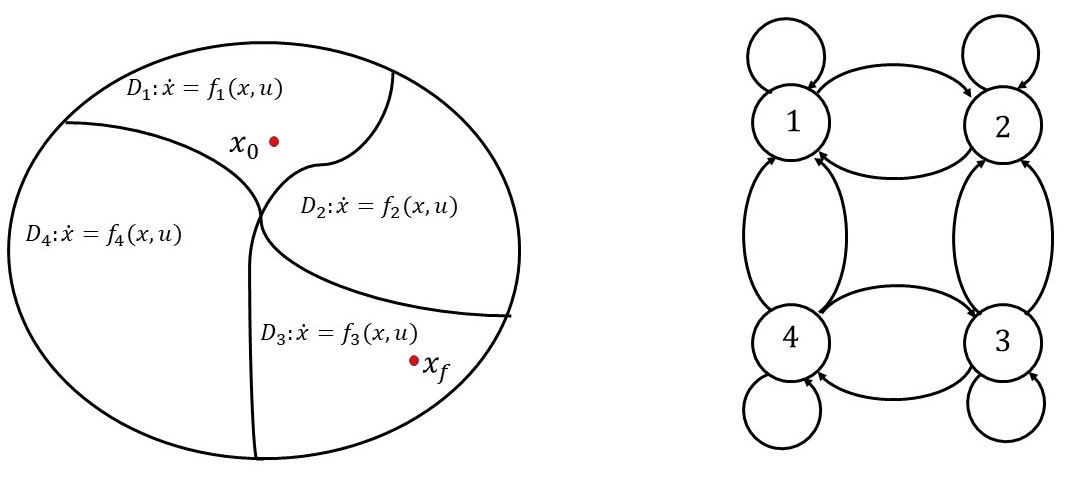}
    \caption{Multiple region dynamical systems: (a) systems illustration; (b) system transition graph.}
    \label{fig:scenario}
\end{figure}

Thus, in this article, we use the reinforcement learning method to solve this problem, which is a sampling-based method and does not rely on prior knowledge of the system model.
This method combines the high-level graph search and low-level optimal control problem into one step to construct the value function in the state space.

Denote the $\mathcal{N}_i$ as the neighborhood region of region $i$.
The formulated problem \eqref{eqn:realsys} and \eqref{eqn:realJ} can be transformed into the following problem:
\begin{align}
\dot x(t) =\sum_{i=1}^m \mathbbold{1}_{x(t)\in \mathcal{R}_i} f_i(x,u)
\end{align}
where $\mathbbold{1}(\cdot)$ denotes the characteristic function.
Similarly, the stage cost is written as
\begin{align}
L= \sum_{i=1}^m \mathbbold{1}_{x(t)\in \mathcal{R}_i} L_i(x,u)
\end{align}

By writing the original problem this way, we can treat the new system as a highly nonlinear system.

\section{Theoretical error analysis} \label{sec:error}
\subsection{Hamilton-Jacobi-Bellman equation for hybrid system}
Define the following value function 
\begin{align*}
V(x(t_f),t)=\min_{u(t)} \int_t^{t_f}L(x(t),u(t))dt +\psi(x(t_f)
\end{align*}
where $V(x(t),t_f)=\psi(x(t_f))$.
Then the Hamilton-Jacobi-Bellman equation satisfies
\begin{align} \label{eqn:HJB}
\nonumber&V(x(t),t) = \\ 
&\min_u \{V(x(t+dt),t+dt)+\int_{t}^{t+dt}L(x(s),u(s))\rm{d}s \}.
\end{align}
Supposing the cost-to-go function is continuously differentiable in $x$ and $t$, we can then do a Taylor expansion of the first term on the right-hand side,
\begin{align} \label{eqn:HJBDDPG}
   \nonumber &V(x(t+dt),t+dt)\\
  \nonumber  &=V(x(t),t)+\frac{\partial V(x,t)}{\partial t}\rm{d}t+\frac{\partial V(x,t)}{\partial x} \cdot \dot {x}(t)\rm{d}t+o(dt) \\
 \nonumber   &=V(x(t),t)+\frac{\partial V(x,t)}{\partial t}\cdot 1\rm{d}t+\frac{\partial V(x,t)}{\partial x} \cdot \dot {x}(t)\rm{d}t+o(dt)\\
 &=V(x(t),t)+\left [\frac{\partial V(x,t)}{\partial x},\frac{\partial V(x,t)}{\partial t} \right 
    ] \cdot \begin{bmatrix}\dot {x}(t)] \\ 1\end{bmatrix}\rm{d}t+o(dt)
\end{align}
where $o(dt)$ denotes terms of higher order than  $dt$.

However, in optimal control for
hybrid systems, we have to account for a non-differentiable value function at some switching states in the switching interface.
Furthermore, the optimal control input $u$ may be discontinuous or non-differentiable at the switching interface as well.
\begin{theorem} \label{thm: hmp}
The value function for above defined optimal control problem for a hybrid system is continuous but may not be differentiable at the switching interface with switching state $x(\tau)$ and switching time instant $\tau$.
\begin{align*}
V(x(\tau_-),\tau_-)=V(x(\tau_+),\tau_+).
\end{align*}
\end{theorem}
\begin{proof}
Based on the dynamic programming formula,
\begin{align*}
V(x(\tau_-),\tau_-)&=V(x(\tau_+),\tau_+)+\int_{\tau_-}^{\tau_+} L(x,u) \mathrm{d}t \\
& \approx V(x(\tau_+),\tau_+)+ \Tilde{L}(x,u)(\tau_+-\tau_-)
\end{align*}
where $\Tilde{L}(x,u)$ is a bounded constant value.
This makes $V(x(\tau_-),\tau_-)=V(x(\tau_+),\tau_+)$ as $\tau_- \rightarrow \tau_+$.
\end{proof}
Recall the definition of Hamiltonian
\begin{align}
H:=\inf_{u \in \mathcal{U}}\; (\lambda f(x,u) + L(x,u))
\end{align}
where $\lambda$ is called co-state.
In the following, we will give two examples where the optimal control inputs are respectively non-differentiable (but continuous) and discontinuous.
\textbf{Example 1: non-differentiable optimal controller $u$:} The following is a simple example to show the nondifferentiability of the optimal controller for a hybrid system.
Consider a one-dimensional switched system:
\begin{align*}
\left\{\begin{matrix}
\dot x = 2u, \quad x>1\\
\dot x = u, \quad x <1 \\
\end{matrix}\right. \\
\min J=\frac{1}{2}\int_0^1 (x^2+u^2)\mathrm{d}t \\
x(0)=2.
\end{align*}
Given the fact that the Hamiltonian is continuous at the switching interface and the optimality conditions in \cite{GELACC_Mi}, we can obtain \footnote{To simplify notation, in this article, both $x_{\pm}$ and $x(\tau \pm)$ means the state at the switching interface. Same for the other variables.}
\begin{align*}
\frac{1}{2}(x_{-}^2+u_{-}^2)+2\lambda_{-}u_{-}&=\frac{1}{2}(x_{+}^2+u_{+}^2)+\lambda_{+}u_{+} \\
\frac{1}{2}x_{-}^2-\frac{1}{2}u_{-}^2&=\frac{1}{2}x_{+}^2-\frac{1}{2}u_{+}^2 \\
&\Rightarrow u_{-}^2=u_{+}^2.
\end{align*}
What's more, $x_-=x_+$, $\lambda_{-}=-\frac{1}{2}u_{-}$, $\lambda_{+}=-u_{+}$.
Based on the Euler-Lagrange equation, we know that, before switching, $\dot \lambda(t)=-x(t) \Rightarrow \dot u(t)=2x(t)$ while after switching $\dot \lambda(t)=-x(t) \Rightarrow \dot u(t)=x(t)$.
Combining these facts, we can derive the continuity but non-differentiability at the switching state of the optimal controller in this example.

\textbf{Example 2: discontinuous optimal controller $u$:} Consider the following first-order state-dependent switched system
\begin{align*}
\dot x = \left\{\begin{matrix}
2x+u, \quad x>1\\
-x+u, \quad x<1
\end{matrix}\right.\\
\min J=\frac{1}{2}\int_0^1 (x^2+u^2)\mathrm{d}t\\
x(0)=2.
\end{align*}
In these two examples, using the optimality conditions derived in \cite{GELACC_Mi} and combined with the algorithm proposed in \cite{jumplaw},
we can find the optimal controllers are respectively non-differentiable and discontinuous at the switching state as shown in Fig. \ref{fig:firstOrderSysJump}.
\begin{figure}[!htp]
    \centering
     \subfigure[]{\includegraphics[width=0.23\textwidth]{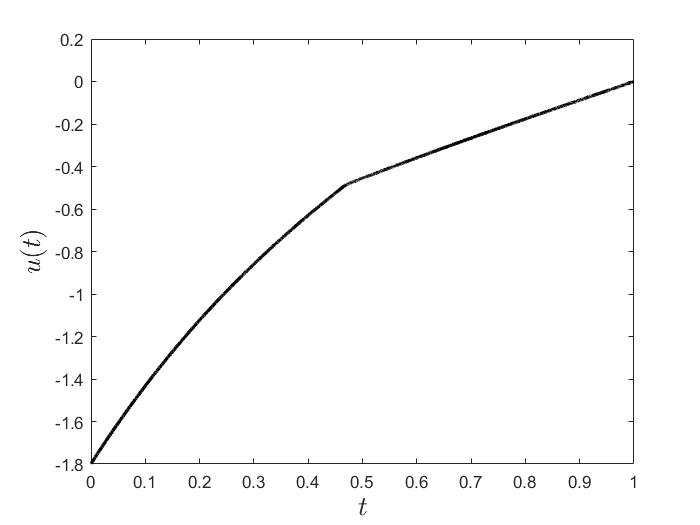}} 
    \subfigure[]{\includegraphics[width=0.23\textwidth]{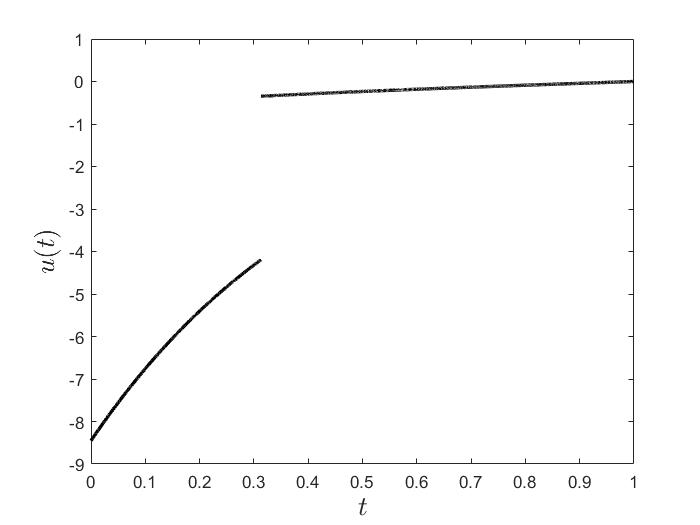}} 
    \caption{(a) Optimal control of Example 1 (switching time $\tau=0.4694$ and optimal cost $J=1.0209$); (b) Example 2 (switching time $\tau=0.3132$ and optimal cost $J=6.5274$).}
    \label{fig:firstOrderSysJump}
\end{figure}

A well-known fact is that the HJB equation \eqref{eqn:HJB} is a partial differential equation, which is notoriously difficult to solve, especially for obtaining the value function.
Thus in this article, we resort to the model-free method called deep deterministic policy gradient in order to solve Eqn. \ref{eqn:HJB}.
In the following, the notation $x$ is the one after augmenting another dimension of time $t$ and $f(x,u) = [\dot x(t),1]^\top$.
\subsection{Approximate dynamic programming and finite horizon RL}
Most recent works about reinforcement learning are about the infinite horizon problem which is related to the linear quadratic regulator in optimal control theory and has a stationary solution.
However, in reality, most optimal control problems have fixed terminal time and we have to optimize an objective function with a fixed terminal time.
This makes the value function a function of state and time, i.e., $V(x,t)$, and thus we define the optimal control policy as $\mu(x,t)$.
The corresponding cost-to-go function can be written as $V^{\mu}(x,t)$.
\begin{theorem}
$u=\mu^*(x,t)$ is the optimal control policy if it minimizes the left-hand-side of the HJB equation \eqref{eqn:HJB} for all $x$ and $t\in [0,t_f]$.
\end{theorem}
\subsection{Approximation theory}
Approximators have been used in different kinds of areas for a long time.
In classical optimal control communities, Chebyshev, Legendre, and Fourier basis functions have been used as approximators for interpolation and a more detailed error-bound analysis can be found in \cite{Cheb_errorBound,errBound2}.
In this section, we analyze the error of the optimal control and the value function $V(x,t)$ given by RL and the real one, i.e., $||u-\hat{u}||$ and $||V(x,t)-\hat{V}(x,t)||$.
We use the approximants $\hat{u}(t)=\theta \phi(x,t)$ and $\hat{V}(x,t)=w\varphi(x,t)$ where $\theta$, and $w$ are the weights and $\phi(\cdot),\; \varphi(\cdot)$ are their corresponding basis function (in this article, it is deep neural networks).
\begin{theorem}[Approximation theory of deep neural network \cite{NNerror}(Theorem 4.16)] \label{thm:NN}
Let $\rho:\mathbb{R}\rightarrow \mathbb{R}$ and $n,m,k\in \mathbb{N}$. Let $\mathcal{NN}_{n,m,k}^\rho$ represent the class of functions $\mathbb{R}^n \rightarrow \mathbb{R}^m$ described by forward neural networks with $n$ neurons in the input layer, $m$ neurons in the output layer, and an arbitrary number of hidden layer, each with $k$ neurons with activation function $\rho$. Every function in the output layer has the identity activation function.
Let $\rho$ be ReLU and $p\in [1,\infty)$. Then $\mathcal{NN}_{n,m,n+m+1}$ is dense \footnote{We say $\rho$ networks are dense in $L^p(\mathcal{X};\mathcal{Y})$ if for any $f^* \in L^p(\mathcal{X}; \mathcal{Y})$ and $\epsilon>0$, there exists a $\rho$ network $f$ such that $||f^*-f||_p \leq \epsilon$.} in $L^p(\mathbb{R}^n;\mathbb{R}^m)$ with respect to the usual $L^p$ norm.
\end{theorem}
With the bounded variation assumption in the control input, we can obtain the convergence of $V(x,t;\theta) \rightarrow V^*(x,t)$ and $u(x,t;\theta) \rightarrow u^*(x,t)$.
In the following, all the norms are $L^p$ norm which is defined as $||f||_p:=(\int_{\mathcal{X}}|f|^p\mathrm{d}x)^{1/p}<\infty,\; p\in [1,\infty)$.
\begin{theorem}
Assume the real value function satisfies the global Lipschitz condition $||V(x(t),t)-V(\hat{x}(t),t)|| \leq \gamma ||x(t)-\hat{x}(t)||$ where $\gamma$ is the Lipschtz constant.
Then there exists a small $\epsilon'>0$ such that
\begin{align}
||V(x(t),t)-\hat{V}(\hat{x}(t),t)|| \leq \epsilon',
\end{align}
where $\hat{V}$ is the value function approximator and $\hat{x}(t)$ is the state with the approximator control input $\hat{u}(t)$ that satisfies the state dynamics Eqn. \eqref{eqn:realsys}.
\end{theorem}
\begin{proof}
 Using Theorem \ref{thm:NN}, we know that there exists $\epsilon=\max \{\epsilon_1, \epsilon_2\}$ such that $||u(t)-\hat{u}(t)|| \leq \epsilon_1 \leq \epsilon$ and $||V(x(t),t)-\hat{V}(x(t),t)||\leq \epsilon_2 \leq \epsilon$.
 Assume the $f_q(x,u)$ is Lipschitz. Then we would have $x(t)-\hat{x}(t)=x(0)-\hat{x}(0)+\int_0^t f(x(s),u(s))-f(\hat{x}(s),\hat{u}(s))\mathrm{d}s$.
 Since the initial state is the same, we have $||x(t)-\hat{x}(t)|| \leq M || u-\hat{u}||_1+M\int_0^t ||x(s)-\hat{x}(s)||\mathrm{d}s$ by H\"{o}lder's inequality.
 Using the Bellman-Gronwall Lemma, we conclude that $||x(t)-\hat{x}(t) ||\leq M||u(t)-\hat{u}(t)||_1e^{Mt}$.
 Thus, using the triangle inequality, we can obtain
 \begin{align*}
&||V(x,t)-\hat{V}(\hat{x},t)|| \leq \\
&||V(x,t)-V(\hat{x},t) || + ||V(\hat{x},t)-\hat{V}(\hat{x},t)|| \\
&\leq \gamma ||x-\hat{x}||+\epsilon \\
&\leq \gamma \beta ||u-\hat{u}||_1+\epsilon \\
&\leq (\gamma\beta+1)\epsilon\\
&\leq \epsilon',
 \end{align*}
 where $\beta=Me^{Mt}$ and $\epsilon' \geq (\gamma \beta+1)\epsilon$.
 As $\epsilon' \rightarrow 0$, $V(x,t)\rightarrow \hat{V}(x,t)$ in $L^p$.
This completes the proof.
\end{proof}
\section{State-Based Multi-agent Deep deterministic policy gradient (SMADDPG)} \label{sec:DDPG}
In this section, we will first introduce briefly the deep deterministic policy gradient algorithm which is a popular actor-critic RL method used to find an optimal policy by interacting with the environment.
Then we will explain why vanilla DDPG algorithm may fail for state-dependent switched systems.
Finally we introduce our proposed SMADDPG algorithm.

DDPG has four neural networks: an actor neural network for parameterizing policy (i.e., control input in our problem), a critic neural network for parameterizing the Q-function ($V(x_t)= \max_{u_t} Q(x_t,u_t)$), and their respective copied target networks to make learning stabilizable.
The parameters of the actor are updated using the policy gradient method as \cite{DDPG}:
\begin{align} \label{eqn:PG}
\nonumber \nabla_{\theta^{\mu_i}} \mu |_{x_t}=\frac{1}{N}\sum_t \nabla_u Q(x,u|\theta^{Q_i})|_{x=x_t, u=\mu(x_t)}\\ \nabla_{\theta ^{\mu_i}}\mu(x|\theta^{\mu_i})|_{x_t},.   
\end{align}
Based on Theorem \ref{thm: hmp}, we know that for systems with regional dynamics, the optimal controller is usually not continuous or differentiable.
Thus the policy gradient estimation step in Eqn. \eqref{eqn:PG} is highly biased \cite{Tedrake}.

To solve this problem, we propose algorithm SMADDPG which uses multiple regional agents to learn the optimal controller for such kind of systems.
Algorithm \ref{alg:SMADDPG} shows the logic flow of the SMADDPG algorithm.
The idea is to put an agent at each region as shown in Fig. \ref{fig:scenario} and make them learn in this environment \footnote{Please note that in this article, we mix the usage of the word ``environment'' and ``system''}.
To simplify the notation, the $x$ also includes the time variable $t$.
The $\mu(x|\theta^\mu)$ is the parameterized optimal control policy and $Q(x_t,u_t)$ is parameterized with parameters $\theta^Q$ and learned using the HJB equation \eqref{eqn:HJB}.
At each episode, we run the environment $T=\frac{t_f}{dt}$ steps where $dt$ is the step size. 
If the current state is in region $i$, we activate the agent $i$ and update its parameters.

\begin{algorithm}
\caption{SMADDPG.}\label{alg:SMADDPG}
\SetKwFunction{Agent}{Agent}
\SetKwFunction{SMADDPG}{SMADDPG}
\Fn{\Agent{state\_dim, action\_dim, current\_state}}{
Select action $u_t=\mu(x_t|\theta^{\mu_i})+\mathcal{N}_t$ according to the current policy and exploration noise\;
Sample a random minibatch of $N$ transitions from $RB$\;
Set $y_t=r_t+ Q'(x_{t+1},\mu'(x_{t+1}|\theta^{\mu'_i})|\theta^{Q_i})$\;
Update critic by minimizing the loss $L=\frac{1}{N}\sum_t (y_t-Q(x_t,u_t|\theta^{Q_i}))^2$\;
Update the actor policy using the sampled gradient based on Eqn. \ref{eqn:PG}
\;
Update the target networks by having them slowly track the learned networks with coefficient $\tau  \ll 1$:
$\theta^{Q_i'}\leftarrow \tau \theta^{Q_i}+(1-\tau)\theta^{Q_i'}$\;
$\theta^{\mu_i'}\leftarrow \tau \theta^{\mu_i}+(1-\tau)\theta^{\mu_i'}$\;
}
\Fn{\SMADDPG}{
Randomly initialize $m$ simple agent with critic networks $Q(x,u|\theta^{Q_i})$ and its corresponding actor networks $\mu(x|\theta^{\mu_i})$ with weights $\theta^{Q_i}$ and $\theta^{\mu_i}$ \;
Initialize target network $Q'$ and $\mu'$ with weights $\theta^{Q'}\leftarrow \theta^Q$, $\theta^{\mu'}\leftarrow \theta^\mu$ where we denote $\theta^\mu=\mathrm{Vec}(\theta^{\mu_i}), i=1,\cdots m$\ and $\theta^Q=\mathrm{Vec}(\theta^{Q_i}), i=1,\cdots m$\;
Initialize $m$ replay buffer $RB_i$ \;
Initialize a random process $\mathcal{N}$ for action exploration\;
Receive initial observation state $x_0$\;
\For{$episode=1,2, \cdots$}{
\For{$t=0,\cdots T-1$}{
\If{$x_t \in \mathcal{R}_i$ }{
$u_t=\mathrm{Agent}_i .\mathrm{act}(x_t)$ \;
 $x_{t+1}, \; r_t = env.step(u_t)$\;
 $\Agent_i.\mathrm{RB}.add(x_t, u_t,x_{t+1},done)$\;
$\Agent_i.\mathrm{update}()$\;
}
}
}
}
\end{algorithm}

In the following section, we will compare our proposed SMADDPG algorithm with vanilla DDPG.
The implementation of the vanilla DDPG algorithm can be found in \cite{DDPG}.
\section{Illustrative examples}\label{sec:simulation}
We test the proposed algorithm in two examples, i.e., a first-order switched system~(FO) and a multiple region switched system~(MR).
The environment is created using first-order Euler method $x_{t+1}=x_t+f(x_t,u_t)dt$. We let $dt=0.01$.
The structure of the DDPG network is as follows: the actor network of DDPG has $a_l$ hidden layers with ReLU activation function and the \texttt{Tanh} output activation and the critic network has $c_l$ hidden layers with \texttt{ReLu} activation functions and \texttt{ReLu} output activation function.
Each example using DDPG is run with 10 different randomly generated seeds under 1000 episodes and the average performance is calculated (written as mean $\pm$ 1 standard deviation).
The hyperparameters are fine tuned to obtain the best results for each algorithm using tensorboard as shown in Table \ref{tab:hyperparams}.
All the simulations are run on a server with GPUs and Pytorch.

\begin{table}[!htp]
    \centering
    \begin{tabular}{c|c|c|c} \\ 
        Environment & Hyper-parameters &  SMADDPG&VDDPG\\ \hline \hline
        \multirow{6}{*}{FO}  & Actor architecture & (32,32,32) & (256,256)\\
        & Critic architecture & (32,32) &   (256,256) \\
        & Batch size & 64& 64 \\
        & actor learning rate & $1e-4$ &$1e-4$  \\
        &critic learning rate & $2e-4$ &$2e-4$ \\
        & noise & 0.05 &0.05\\ \hline
        \multirow{6}{*}{MR}  & Actor architecture & (256,256) & (256,256,256)\\
        & Critic architecture & (256,256) &  (256,256)\\
        & Batch size &  64& 64 \\
        & actor learning rate & $1e-4$  & $1e-4$\\
        &critic learning rate & $2e-4$  & $2e-4$\\
        & noise & 0.05 &0.05\\ \hline
    \end{tabular}
    \vspace{4pt}
    \caption{Best hyperparameters for different environments.}
    \label{tab:hyperparams}
\end{table}
\subsection{Example 1: FO}
The first example is a first-order state-dependent switched system. 
\begin{align*}
\min J=\frac{1}{2}\int_0^1 (x^2+u^2)\mathrm{d}t\\
s.t., \; \dot x = \left\{\begin{matrix}
2x+u, \quad x>1\\
-x+u, \quad x<1
\end{matrix}\right.\\
x(0)=2.
\end{align*}
where the fixed terminal time $t_f=1$.
The system's dynamic switches at the state $x=1$.
The reward function is designed as $r=-0.5(x_t^2+u_t^2)$.

Fig. \ref{fig:FOcompare}(a) shows the reward curve with fixed seed and fine-tuned hyper-parameters.
As we can see, SMADDPG has a higher reward and the control policy (Fig. \ref{fig:FOcompare}(b)) learned to switch at some time instant $t\in[0.2,0.4]$ which is consistent with Fig.~\ref{fig:firstOrderSysJump}(b).
The average performance of SMADDPG is better than that of vanila DDPG as well with lower variance during the learning process as shown in Fig. \ref{fig:FOcompare}(c).
\begin{figure*}[!htp]
    \centering
     \subfigure[]{\includegraphics[width=0.32\textwidth]{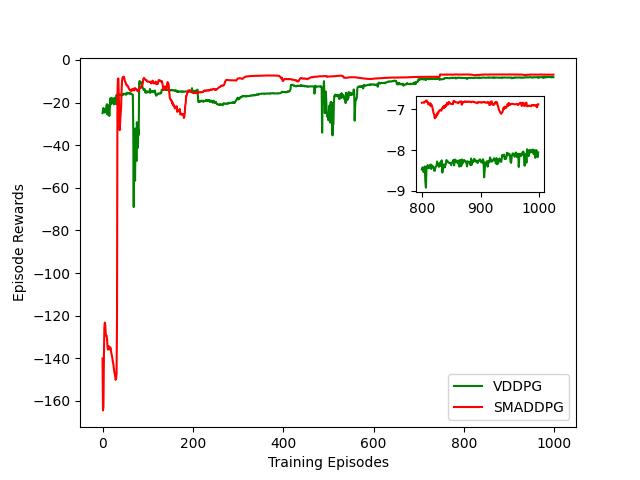}} 
    \subfigure[]{\includegraphics[width=0.32\textwidth]{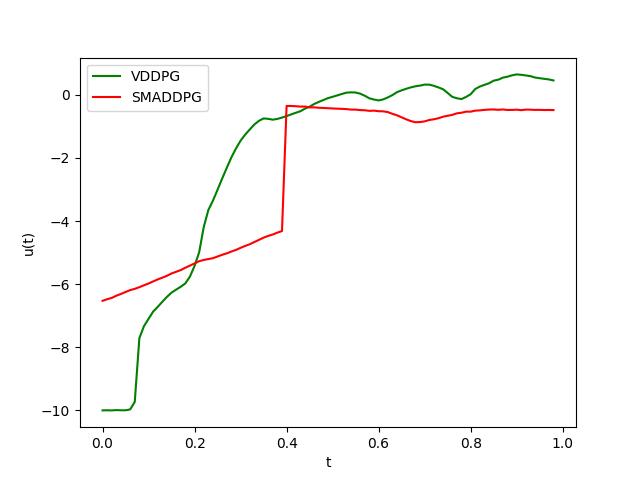}} 
    \subfigure[]{\includegraphics[width=0.32\textwidth]{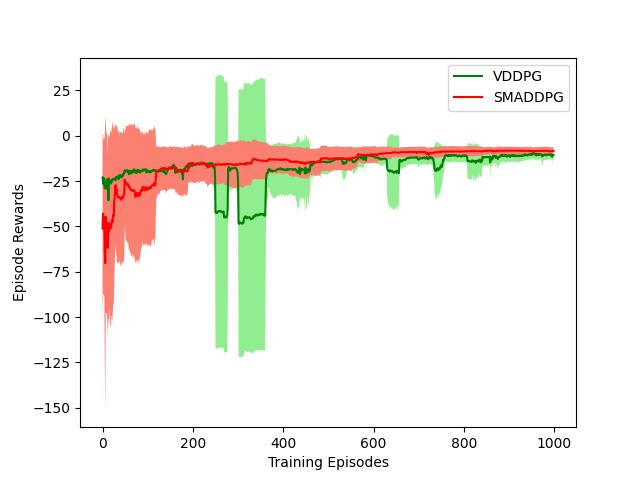}}
    \caption{\textbf{FO}: (a) Episode reward (fixed seed); (b) Learned control policy (fixed seed); (c) Average episodic reward under 10 runs with randomly generated seed.}
    \label{fig:FOcompare}
\end{figure*}

\subsection{Example 2: MR}
Next, we consider the following example adapted from \cite{MR}. This case also satisfies our assumptions:
\begin{align*}
\dot x = A_q x +B_q u
\end{align*}
with
\begin{align*}
&A_1 = \begin{bmatrix}
-1 & 2 \\
-2 & -1 \\
\end{bmatrix}, A_2 = \begin{bmatrix}
-1 & -2 \\
1 & -5 \\
\end{bmatrix}, A_3=\begin{bmatrix}
-0.5 & -5 \\
1 & -0.5 \\
\end{bmatrix}, \\&
A_4 = \begin{bmatrix}
-1 & 0 \\
2 & -1 \\
\end{bmatrix}, B=\begin{bmatrix}
1 \\
1
\end{bmatrix}.    
\end{align*}
The stage cost function is $L_i = \frac{1}{2}(x^\top x+ u^2),\; i=1,2,3,4$, and the terminal time is $t_f = 2$. 
The initial state is $x_0 = (-8,-6)^\top$.
Separating regions $\mathcal{R}_i, i=1,2,3,4$, the switching interfaces are $m_{12} = x_2+5 = 0$, $m_{13} = x_1+5=0$, $m_{23}=-m_{32}=x_1-x_2=0$, $m_{24}=-m_{42}=x_1+2=0$, $m_{34}=-m_{43}=x_2+2=0$.
The controller is constrained in all locations to $-10\leq u_q \leq 10$.

For the DDPG algorithm, we set up the environment as a state-dependent switched system.
The observation space is the position $(x_1,x_2,t)$ where $x_1\in [-10,10],\; x_2\in [-10,10], \; t\in [0,2]$.
The action space is the control input $u\in [-10,10]$.
The reward function is defined as $r_t=-L_i(x_t,u_t)$.

Fig. \ref{fig:MRcompare}(a)-(b) are the episode reward curve and corresponding learned control policy after convergence obtained by SMADDPG and VDDPG respectively after fine tuning the hyper-parameters.
It shows a higher reward obtained and more stable convergence result.
Fig. \ref{fig:MRcompare}(c) is the episode reward under 10 trials with different random seed.
Lower variance is observed for SMADDPG.
\begin{figure*}[!htp]
    \centering
     \subfigure[]{\includegraphics[width=0.32\textwidth]{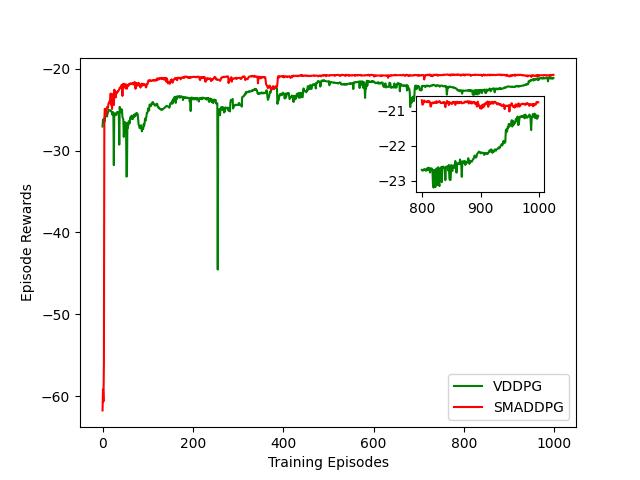}} 
    \subfigure[]{\includegraphics[width=0.32\textwidth]{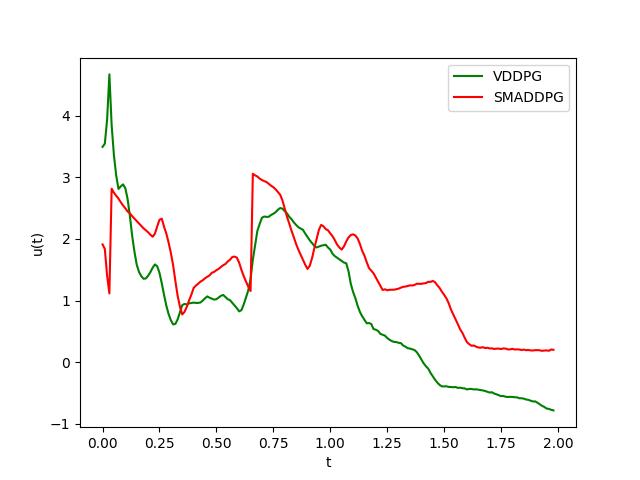}} 
    \subfigure[]{\includegraphics[width=0.32\textwidth]{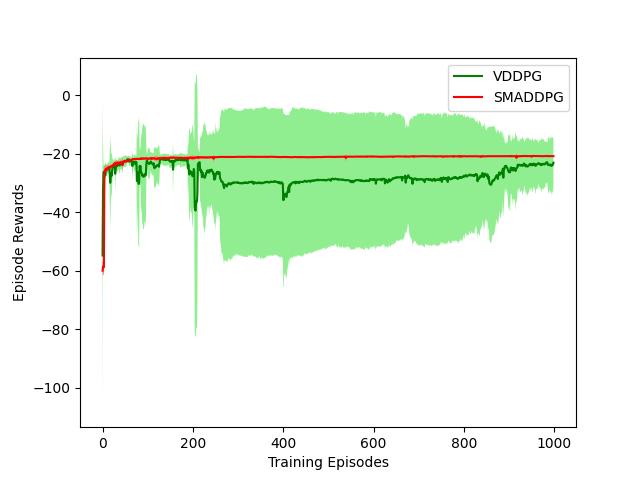}}
    \caption{\textbf{MR}: (a) Episode reward (fixed seed); (b) Learned control policy (fixed seed); (c) Average episodic reward under 10 runs with randomly generated seed.}
    \label{fig:MRcompare}
\end{figure*}

\section{Conclusion} \label{sec:conclusion}
In this article, we proposed to use multiple DDPG agents to solve optimal control for state-dependent switched systems.
We theoretically analyzed the feasibility of the proposed idea and verified it in two customized state-dependent switched systems.
We observed an improved performance of proposed method compared to vanilla DDPG.
The present approach focuses on the case where the switching interface is known while the actual system dynamics itself remains unknown. This simplifying assumption allowed us to demonstrate and solve the hidden problems of learning optimal control policy for hybrid systems. 
Our future work will address the case where the switching interface is treated as unknown along with optimal control.

\bibliographystyle{IEEEtran}
\bibliography{IEEEabrv,IEEEexample}

\end{document}